\documentclass[12pt,reqno]{amsart}
\usepackage{dsfont, amssymb,amsmath,amscd,latexsym, amsthm, amsxtra,amsfonts}
\usepackage[all]{xy}
\usepackage{graphicx}
\usepackage{float}
\usepackage[active]{srcltx}
\newtheorem{theorem}{Theorem}[section]

\newtheorem{example}[theorem]{Example}

\newtheorem{Them}{Theorem}[section]

\begin{document}
\makeatletter
\def\@setauthors{%
\begingroup
\def\thanks{\protect\thanks@warning}%
\trivlist \centering\footnotesize \@topsep30\p@\relax
\advance\@topsep by -\baselineskip
\item\relax
\author@andify\authors
\def\\{\protect\linebreak}%
{\authors}%
\ifx\@empty\contribs \else ,\penalty-3 \space \@setcontribs
\@closetoccontribs \fi
\endtrivlist
\endgroup } \makeatother
 \baselineskip 17pt
\title[{{ OPTIMAL  DIVIDEND and reinsurance CONTROL }}]
 {{ Optimal Dividend and reinsurance strategy of a Property Insurance
 Company under Catastrophe Risk}}
\author[{\bf   Zongxia Liang, Lin He  and Jiaoling  Wu  } ]
{ Zongxia Liang \\ Department of Mathematical Sciences, Tsinghua
University, Beijing, China\\ Email: zliang@math.tsinghua.edu.cn\\
Lin He \\ The School of Finance, Renmin University of China,
Beijing, China \\Email: helin@ruc.edu.cn\\
Jiaoling  Wu\\
Department of Mathematical Sciences, Tsinghua University, Beijing,
China\\ Email: maths2005ling@hotmail.com } \maketitle
\begin{abstract}
We consider an optimal control problem of a property insurance
company with proportional reinsurance strategy. The insurance
business brings in catastrophe risk, such as earthquake and flood.
The catastrophe risk  could be partly reduced by reinsurance. The
management of the company controls the reinsurance rate and dividend
payments process to maximize the expected present value of the
dividends before bankruptcy. This is the first time to consider the
catastrophe risk in property insurance model, which is more
realistic.  We  establish the solution of the problem by the mixed
singular-regular control of jump diffusions. We first derive the
optimal retention ratio, the optimal dividend payments level, the
optimal return function and the optimal control strategy of the
property insurance company, then the impacts  of  the catastrophe
risk and key model parameters on the optimal return function and the
optimal control strategy of the company are discussed. \vskip 10 pt
 \noindent {\bf MSC}(2000): Primary 91B30, 91B70, 91B28; Secondary  60H10, 60H30.
\end{abstract}
 \vskip 10pt  \noindent
\setcounter{equation}{0}
\section{{ {\bf Introduction}}}
 \vskip 15pt \noindent
 In this paper we consider a property insurance company in which the
 dividend payments process and risk exposure are controlled by the management.
 The property insurance business brings in catastrophe risk, such as
 earthquake and flood.
 We assume that the company can only reduce its risk exposure by
 proportional insurance strategy for simplicity. The catastrophe
risk  could also be partly reduced by reinsurance. The regulation of
the catastrophe risk determines to what extent the catastrophe risk
could be eliminated, here we use reinsurance rate and adjusted risk
rate  in the regulation. We equate the value of the company to the
expected present value of the dividend payments before bankruptcy.
 \vskip 10pt \noindent
This is a mixed singular-regular control on diffusion models with
jumps. These optimization problems of  diffusion models for property
insurance companies that control their risk exposure by means of
dividend payments have attracted  significant interests recently. We
refer readers  to Radner and Shepp \cite{s18}, Paulsen and Gjessing
\cite{s15}, H{\o}jgaard and Taksar\cite{s3,s4} and
Asmussen\cite{s8}. Optimizing  dividend payments is a classical
problem  starting from the early work of Borch\cite{s10,s9},
Gerber\cite{s11}. For some applications of control theory in
insurance mathematics, see H{\o}jgaard and Taksar\cite{s12,s13},
Martin-l\"{o}f\cite{s14}, Asmussen and Taksar\cite{s7,s6} and  He
and Liang\cite{s200, s201, ime02}, Basse, Reddemann, Riegler and
Schulenburg\cite{BRJS}, Guo, Liu and Zhou \cite{s1} and other
author's work. Recent surveys can  be found in Taksar\cite{s16},
Avanzi \cite{AV}, Albrecher and Thonhauser\cite{AT}.
 \vskip 10pt\noindent
 Unfortunately, there is little work concerned with
 the catastrophe risk of the property insurance company in the problem of
 optimal risk control/dividend distribution via the reinsurance rate.
In the real financial market, the property insurance business
generally brings in catastrophe risk, such as earthquake and flood.
 The asset of the company evolves as a l\'{e}vy process
  with jump diffusions. Harrison
   and Taksar \cite{s101} provides a good idea to solve this
   kind of problems. Bernt {\O}ksendal and Agn\`{e}s
   Sulem\cite{s202} study the stochastic control problem of jump
   diffusions. Enlightened by these innovative
   ideas, we can solve effectively the optimal control problem of the
   company under catastrophe risk.
   Firstly, we establish the control problem of the L\'{e}vy processes
   with jump diffusions which is a realistic model of
   the property insurance company facing catastrophe risk. Then we work out
the solution of    singular-regular control of the jump diffusions,
that is, we establish the optimal return function, the optimal
reinsurance rate and  the optimal dividend strategy of the insurance
company. Finally we study the impacts of some key model parameters
on the optimal return function and the optimal dividend strategy.
 \vskip 10pt\noindent
The paper is organized as follows: In next section, we establish the
mathematical control model of the insurance company facing
catastrophe risk.  In section 3, we work out a solution of HJB
equations associated with  the singular-regular control on L\'{e}vy
processes with  jump diffusions.  In section 4, we establish the
solution of the optimal control problem, i.e., we derive the optimal
return function, the optimal reinsurance rate and the optimal
dividend strategy of the property insurance company. In section 5,
we use numerical calculations to discuss the influences of  the key
model parameters on the optimal retention ratio, the optimal
dividend payments level, the optimal return function and optimal
control strategy of the company. In section 6, we summarize main
results of this paper.
\vskip 15pt\noindent
\setcounter{equation}{0}
\section{{ {\bf Mathematical model with proportional reinsurance strategy
under catastrophe risk }}} \vskip 15pt\noindent In this paper, we
consider a property insurance company with proportional reinsurance
strategy. The property insurance business brings in catastrophe
risk, such as earthquake and flood. The catastrophe risk could only
be partly reduced by reinsurance. The company's management can
accommodate the profit and the risk by choosing dividend payments
process and  reinsurance rate. \vskip 8pt\noindent The asset of the
company evolves as the L\'{e}vy processes with jump diffusions. In
this model, if there is no dividend payments and only the
proportional reinsurance strategy is used to control the risk, then
the asset of the property insurance company is approximated by the
following processes(see {\O}ksendal and Sulem\cite{s202}),
\begin{eqnarray*}
dR_{t}=\mu a(t)dt+\sigma a(t)d{W}_{t}+ka(t)\int_\Re
z\widetilde{N}(dt,dz),
\end{eqnarray*}
where ${W}_{t}$ is a standard Brownian motion, $ \mu  $ is the
premium rate, and $\sigma^2 $ is the volatility rate, it is a normal
description of the property insurance company. $1-a(t)\in[0,1]$ is
the proportional reinsurance rate. $\widetilde{N}(dt,dz)=N(dt,
dz)-{I}_{\{|z|<R\}}\nu(dz)dt$ is the compensated Poisson random
measure of L\'{e}vy process $\{N_t\}$ with finite L\'{e}vy measure
$\nu$. The jump diffusions stand for the catastrophe risk produced
by earthquake and flood in the property insurance business. The
catastrophe risk could be partly reduced by reinsurance strategy.
Since the catastrophe risk is huge, the reinsurance strategy is not
the same as the normal reinsurance. Denote $k$ as the adjusted risk
rate according to the reinsurance regulation of the catastrophe
risk. $k$ is a constant. Throughout this paper we assume that $k\in
(0,\frac{\mu}{2\int_\Re z\nu(dz)} ]$, which ensures that the company
does not go into bankruptcy as soon as the catastrophe risk appears.
 \vskip 10pt\noindent
To give a mathematical foundation of the optimization problem, we
fixed a filtered probability space $(\Omega, \mathcal{F} , \mathcal
{F}_{t}, P)$, $\{{W}_{t}\}$ is a standard Brownian motion,
$\widetilde{N}(dt,dz)=N(dt, dz)-I_{\{|z|<R\}}\nu(dz)dt$ is also the
compensated Poisson random measure of L\'{e}vy process $\{N_t\}$
with finite L\'{e}vy measure $\nu$ on  this probability space.
 $\mathcal {F}_{t}$ represents the information available
at time $t$ and any decision is made based on this information. In
our model, we denote $L_{t}$ as the cumulative amount of dividend
payments from time $0$ to time $t$. We assume  that the dividend
payments  process $L_{t}$ is an $\mathcal {F}_{t}$ -adapted,
non-decreasing and right-continuous with left limits. \vskip
10pt\noindent
 A control strategy $ \pi $ is
described by a pair of $\mathcal {F}_{t}$ -adapted stochastic
processes $ \{ a_\pi, L^\pi\}$.
 A strategy $\pi =\{a_\pi (t),L_t
^\pi\} $ is called admissible if $0 \leq a_\pi (t)\leq 1 $ and $L_t
^\pi$ is a nonnegative, non-decreasing and right-continuous
function. We denote $\Pi$ the set of all admissible policies. When a
admissible strategy $\pi
 $ is applied,  we can rewrite  the asset of the
insurance company  by the following processes,
\begin{eqnarray*}
dR^\pi_{t}=\mu a_\pi(t)dt+\sigma a_\pi(t)d{W}_{t}+ka_\pi(t)\int_\Re
z\widetilde{N}(dt,dz)-dL^\pi_{t}, \quad R^\pi_{0}=x.
\end{eqnarray*}
In this case, we consider  transaction cost in the dividend
procedures. To simplify the problem, we consider the proportional
transaction cost, that is, if the company pays $l$, as dividend
payments, then the shareholders can get $\beta l,\beta <1$. The
company is considered bankruptcy as soon as its asset falls below
$0$. We define the bankrupt time as $\tau_\pi=\inf\{t\geq0:
R^\pi_{t}\leq 0\}$. $\tau_\pi$ is clearly an $\mathcal {F}_{t}$
-stopping time. \vskip 10pt\noindent The performance function
associated with each $\pi$
 is defined by
\begin{eqnarray}\label{e2.1}
J(s,x, \pi )&=&\mathbf{E}\big[\int_{0}^{\tau_\pi}e^{-c(s+t)}\beta
dL^\pi_{t}\big ],
\end{eqnarray}
and the optimal return function is
\begin{eqnarray}\label{e217}
 V(s,x)&=&\sup\limits_{\pi\in \Pi }\big\{ J(s, x,\pi)\big\},
\end{eqnarray}
where  $c$ denotes the discount rate. If a strategy $\pi^* $ is such
that $ J(s,x, \pi^* )= V(s,x)$, then we call $\pi^*  $,  $ a_{\pi^*}
(t)$ and $L^{\pi^*}_{t}$ the optimal dividend strategy, the optimal
retention ratio and the optimal dividend payments process,
respectively. This paper aims at working out the optimal strategy as
well as the optimal return function, and then discussing impacts of
key model parameters(e.g. $k$, $\nu$, $\mu$ and $\sigma^2$) on
$V(s,x)$, $ a_{\pi^*} (t)$ and $L^{\pi^*}_{t}$. \vskip 10pt\noindent
 \setcounter{equation}{0}
\section{{{\bf The solution of HJB equations for(\ref{e2.1}) and (\ref{e217})  }} }
\vskip 10pt\noindent In order to solve the  optimal stochastic
control problem (\ref{e2.1}) and (\ref{e217}) of jump diffusions in
next section,  we establish a solution of HJB equation associated
with the control problem in this section. The main result of this
section is the following.
\begin{Them} Assume that the L\'{e}vy measure $\nu$  and the
adjusted risk rate  $k$ satisfy $0<  \nu( \Re)< +\infty$,
$0<\int_\Re z\nu(dz)< +\infty  $ and $ 0< k \leq
\frac{\mu}{2\int_\Re z\nu(dz) }$. Let $\phi(s, x)$ be the function
defined by
 $$ \phi(s, x)= e^{-cs}\psi (x)\ \mbox{and}  $$
\begin{eqnarray}\label{eq200}
\psi (x)=\left\{
\begin{array}{l l l}
\psi_1(x)=C_1x^{\gamma},\ 0\leq x\leq x_0,  \\
\psi_2(x)=C_3e^{d_-x}+C_4e^{d_+x}, \ \ x_0 \leq x\leq x^* ,\\
\psi_3(x)=\beta (x-x^*)+\psi_2(x^*), \ \   x\geq x^*,
\end{array}\right.
\end{eqnarray}
 where $x_0
 =\frac{(1-\gamma)\sigma^2}{\mu}$,
  $\gamma$, $d_-$ and $d_+$ are solutions of (\ref{e310}) and
(\ref{e315}) below with
\begin{eqnarray}\label{relat}
\frac{x_0}{\gamma}+\frac{1}{|d_-|}-\frac{1}{d_+}<0.
\end{eqnarray}
 $x^*$, $C_1$, $C_2$ and $C_3$ are determined by  (\ref{x^*}), (\ref{C1}),( \ref{C3})
and (\ref{C4}) below, respectively. Then $\phi(s, x)\in C^2$ and is
a solution of the following HJB equation
\begin{eqnarray}\label{e32}
\max\big \{-\frac{\partial\phi}{\partial x}(s,x)+\beta e^{-cs},
\max\limits_{a\in[0,1]}\{\mathcal {A}\phi\} \big\}=0,
\end{eqnarray}
where
\begin{eqnarray*}
\mathcal {A}\phi&=&\frac{\partial\phi}{\partial
s}+\frac{\partial\phi}{\partial
x}a\mu+\frac{1}{2}a^2\sigma^2\frac{\partial^2\phi}{\partial x^2}
+\int_{\Re}\big\{\phi(s,x+akz)-\phi(s,x)\nonumber\\&-&akz\frac{\partial\phi}{\partial
x}(s,x)\big\}\nu(dz).
\end{eqnarray*}
\end{Them}
\vskip 5pt\noindent
\begin{proof}  Define $D$ as
\begin{eqnarray*}
D=\{(s,x):-\frac{\partial\phi}{\partial x}(s,x)+\beta e^{-cs}<0\}.
\end{eqnarray*}
We guess that
\begin{eqnarray*}
D=\{(s,x):s\geq 0,\ \  0<x<x^*\}
\end{eqnarray*}
for some unidentified $x^*$.  \ \  Inside $D$, the $\phi$ satisfies
\begin{eqnarray}\label{e34}
\max\limits_{a\in[0,1]}\{\mathcal {A}\phi\}=0,
\end{eqnarray}
i.e.,
\begin{eqnarray}\label{e35}
\max\limits_{a\in[0,1]}\big\{\frac{\partial\phi}{\partial
s}&+&\frac{\partial\phi}{\partial
x}a\mu+\frac{1}{2}a^2\sigma^2\frac{\partial^2\phi}{\partial x^2}
+\int_{\Re}\{\phi(s,x+akz)-\phi(s,x)\nonumber\\&-&akz\frac{\partial\phi}{\partial
x}(s,x)\}\nu(dz)\big\}=0.\nonumber\\
\end{eqnarray}
Differentiating $\mathcal {A}\phi =0 $ w.r.t. $a$, we get
\begin{eqnarray}\label{e36}
\frac{\partial\phi}{\partial
x}\mu+a\sigma^2\frac{\partial^2\phi}{\partial x^2}=0.
\end{eqnarray}
The equation (\ref{e36}) implies that the maximizer of the
right-hand side of the equation (\ref{e35}), $a(x)$, is the
following
\begin{eqnarray}\label{e37}
a(x)=-\frac{\mu\frac{\partial\phi}{\partial
x}}{\sigma^2\frac{\partial^2\phi}{\partial x^2}}.
\end{eqnarray}
Putting the expression  (\ref{e37}) into the equation (\ref{e35}),
we derive
\begin{eqnarray}\label{e38}
\frac{\partial\phi}{\partial
s}&-&\frac{1}{2}\frac{\mu^2(\frac{\partial\phi}{\partial
x})^2}{\sigma^2\frac{\partial^2\phi}{\partial
x^2}}+\int_{\Re}\{\phi(s,x-\frac{\mu\frac{\partial\phi}{\partial
x}}{\sigma^2\frac{\partial^2\phi}{\partial
x^2}}kz)-\phi(s,x)\nonumber\\&+&\frac{\mu\frac{\partial\phi}{\partial
x}}{\sigma^2\frac{\partial^2\phi}{\partial
x^2}}kz\frac{\partial\phi}{\partial x}(s,x)\}\nu(dz)=0.\nonumber\\
\end{eqnarray}
Define $\phi=e^{-cs}\psi(x)$, then it is easy to see from
(\ref{e38}) that the function $\psi(x)  $ satisfies
\begin{eqnarray}\label{e39}
-c\psi-\frac{1}{2}\frac{\mu^2(\psi^{'})^2}{\sigma^2\psi^{''}}
+\int_{\Re}\{\psi(x-\frac{\mu\psi^{'}}{\sigma^2\psi''}kz)-
\psi+\frac{\mu(\psi^{'})^2}{\sigma^2\psi''}kz\}\nu(dz)=0.\nonumber\\
\end{eqnarray}
Because $ a(x)\in [0,1)$, $0\leq x \leq x_0 $ and $a(x)=1 $, $x\geq
x_0$ for some $x_0\geq 0 $, we guess that  $
\psi(x)=\psi_1(x):=C_1x^\gamma +C_2$, $0\leq x \leq x_0 $. Using
$\psi(0)=0$, we have $ \psi(x)=C_1x^\gamma$. Putting it into
(\ref{e39}), we derive the following equation
\begin{eqnarray}\label{e310}
-c-\frac{1}{2}\frac{\mu^2}{\sigma^2}\frac{\gamma}{\gamma-1}
+\int_\Re \{ ( 1-\frac{\mu}{\sigma^2}\frac{1}{\gamma-1}kz)^\gamma-1
+\frac
{\mu}{\sigma^2}\frac{\gamma}{\gamma-1}kz\}\nu(dz)=0.\nonumber\\
\end{eqnarray}
By the assumption of L\'{e}vy measure $\nu $  every term in the
(\ref{e310}) is well-defined. Let $h(\gamma)$ denote the left hand
side of the (\ref{e310}). Then by the assumption of $k$  we have
$h(1_-):=\lim\limits_{\gamma< 1, \gamma \rightarrow 1}\{h(\gamma)\}=
+\infty $ and $h(0)=-c <0 $. So there is at least a $ \gamma $ to
solve the equation (\ref{e310}). Thus $ \psi_1(x)=C_1 x^{\gamma}$
and $a(x)=\frac {\mu x}{\sigma^2(1-\gamma )}$ for $0\leq x\leq x_0
 =\frac{(1-\gamma)\sigma^2}{\mu}$ because of $ a(x)\in [0,1]$.
 \vskip 10pt\noindent
 If $x_0\leq x\leq x^* $, then $a(x)=1$ and the
(\ref{e35} ) becomes
\begin{eqnarray}\label{e313}
\frac{\partial\phi}{\partial s}+\frac{\partial\phi}{\partial
x}\mu+\frac{1}{2}\sigma^2\frac{\partial^2\phi}{\partial x^2}
+\int_{\Re}\{\phi(s,x+kz)-\phi(s,x)-kz\frac{\partial\phi}{\partial
x}(s,x)\}\nu(dz)=0.\nonumber\\
\end{eqnarray}
Define $\phi(x)=\phi_2(x):=e^{-cs}\psi_2(x)$ for $x_0\leq x \leq
x^*$, then we derive from the (\ref{e313}) that
\begin{eqnarray}\label{e314}
\frac{1}{2}\sigma^2\psi_2^{''}(x) +\mu \psi_{2}^{'}(x)-c\psi_2(x)
+\int_{\Re}\{\psi_2(x+kz)-\psi_2(x)-kz \psi_2^{'}(x) \}\nu(dz)=0.\nonumber\\
\end{eqnarray}
We guess that
$$\psi_2(x)=e^{dx}\  \mbox{ for some constant $d\in \Re $   }$$
and further get the equation
\begin{eqnarray}\label{e315}
l(d):=\frac{1}{2}\sigma^2d^2+\mu d -c +\int_\Re \big\{e^{kdz}-1-kdz
\big\}\nu(dz)=0.
\end{eqnarray}
Since $l(0)<0$ and
$\lim\limits_{d\rightarrow+\infty}l(d)=\lim\limits_{d\rightarrow-\infty}l(d)=+\infty$,
the equation(\ref{e315}) has two solutions $d_{-}$ and $d_{+}$ with
$d_{-}<0<d_{+}$, and so the $\psi_2(x)$ should have the following
form
\begin{eqnarray*}
\psi_2(x)=C_3e^{d_{-}x}+C_4e^{d_{+}x} \ \mbox{ for $ x_{0}\leq x\leq
x^{*}$}
\end{eqnarray*}
where $ C_3$ and $C_4$ are constants.
 \vskip 10pt\noindent
 For $x\geq x^{*}$, the solution $\phi=e^{-cs}\psi_3(x)$ and
\begin{eqnarray*}
\psi_3(x)=\beta(x-x^{*})+\psi_2(x^{*})\ \mbox{ for  $x\geq x^{*}$}.
\end{eqnarray*}
Since $\psi'$ and $\psi''$ are continuous at $x^*$,
\begin{eqnarray}\label{e202}
\psi_2^{'}(x^{*})=\psi_3^{'}(x^{*}),
\end{eqnarray}
\begin{eqnarray}\label{e203}
\psi_2^{''}(x^{*})=\psi_3^{''}(x^{*}).
\end{eqnarray}
So
\begin{eqnarray*}
C_3(x^*)d_{-}e^{d_{-}x^*}+C_4(x^*)d_{+}e^{d_{+}x^*}=\beta,\\
C_3(x^*)d^2_{-}e^{d_{-}x^*}+C_4(x^*)d^2_{+}e^{d_{+}x^*}=0.
\end{eqnarray*}
Solving the last two equations, we have
\begin{eqnarray}\label{C3}
C_3(x^*)&=&\frac{\beta d_{+}}{e^{d_{-}x^*}d_{-}(d_{+}-d_{-})}<0,\\
C_4(x^*)&=&\frac{\beta
d_{-}}{e^{d_{+}x^*}d_{+}(d_{-}-d_{+})}>0.\label{C4}
\end{eqnarray}
Also, since $\psi$ and $\psi^{'}$ are continuous at $x_0$,
\begin{eqnarray*}
\psi_1(x_0)=\psi_2(x_0),\\
\psi_1^{'}(x_0)=\psi_2^{'}(x_0),
\end{eqnarray*}
that is,
\begin{eqnarray}\label{e204}
C_1x_0^\gamma=C_3(x^*)e^{d_{-}x_0}+C_4(x^*)e^{d_{+}x_0},
\end{eqnarray}
\begin{eqnarray}\label{e205}
C_1\gamma
x_0^{\gamma-1}=C_3(x^*)d_{-}e^{d_{-}x_0}+C_4(x^*)d_{+}e^{d_{+}x_0}.
\end{eqnarray}
We deduce from  the equations (\ref{e204}) and (\ref{e205}) that
\begin{eqnarray*}\label{e206}
q(x^*)&:=&(\frac{x_0}{\gamma}-\frac{1}{d_{-}})\frac{\beta
d_{+}}{(d_{+}-d_{-})}e^{d_{-}(x_0-x^*)}\nonumber\\
&&-(\frac{x_0}{\gamma}-\frac{1}{d_{+}})\frac{\beta
d_{-}}{(d_{+}-d_{-})}e^{d_{+}(x_0-x^*)}=0.
\end{eqnarray*}
We claim that the $x^*$ satisfying the last equation does exist. In
fact, differentiating $q(x)$, we have
\begin{eqnarray*}
q^{'}(x)=-(\frac{x_0}{\gamma}d_{-}-1)\frac{\beta
d_{+}}{(d_{+}-d_{-})}e^{d_{-}(x_0-x)}-(\frac{x_0}{\gamma}d_{+}-1)\frac{\beta
d_{-}}{(d_{-}-d_{+})}e^{d_{+}(x_0-x)}\nonumber\\
=-\beta\{\frac{x_0d_{+}d_{-}}
{\gamma(d_{+}-d_{-})}(e^{d_{-}(x_0-x)}-e^{d_{+}(x_0-x)})
+\frac{d_{-}e^{d_{+}(x_0-x)}-d_{+}e^{d_{-}(x_0-x)}}{d_{+}-d_{-}}\}>0
\end{eqnarray*}
for $x>x_0$. So $q(x)$ is an increasing function of $x$ and reaches
its minimum at $x_0$. Furthermore, by (\ref{relat}) we have
\begin{eqnarray}\label{e207}
q(x_0)=(\frac{x_0}{\gamma}-\frac{1}{d_{-}})\frac{\beta
d_{+}}{(d_{+}-d_{-})}-(\frac{x_0}{\gamma}-\frac{1}{d_{+}})\frac{\beta
d_{-}}{(d_{+}-d_{-})}<0.
\end{eqnarray}
Also $\lim\limits_{x\rightarrow+\infty}q(x)=+\infty$. Thus there
exists an $x^{*}(>x_0)$ satisfying $q(x^*)=0$. Solving the equation
$q(x^*)=0$, we get
\begin{eqnarray}\label{x^*}
x^*=x_0-\frac{1}{d_+-d_-}\ln \big\{
\frac{d_+^2(d_-x_0-\gamma)}{d^2_-(d_+x_0-\gamma )}\big\}.
\end{eqnarray}
Clearly, (\ref{relat}) implies that $0<
\frac{d_+^2(d_-x_0-\gamma)}{d^2_-(d_+x_0-\gamma )}<1, $  so $x^*>
x_0$. Moreover,
\begin{eqnarray}\label{C1}
C_1(x^*)=\frac{\beta d_{+}}{x_0^\gamma
e^{d_{-}x^*}d_{-}(d_{+}-d_{-})}e^{d_{-}x_0}+\frac{\beta
d_{-}}{x_0^\gamma
e^{d_{+}x^*}d_{+}(d_{-}-d_{+})}e^{d_{+}x_0}>0.\nonumber\\
\end{eqnarray}
Therefore the function $\phi(s, x) $ defined by the (\ref{e32})
should be the following form
 $$ \phi(s, x)= e^{-cs}\psi (x) \ \mbox{ and}  $$
\begin{eqnarray}\label{eq316}
\psi (x)=\left\{
\begin{array}{l l l}
\psi_1(x)=C_1(x^*)x^{\gamma},\ 0\leq x\leq x_0,  \\
\psi_2(x)=C_3(x^*)e^{d_-x}+C_4(x^*)e^{d_+x}, \ \ x_0 \leq x\leq x^* ,\\
\psi_3(x)=\beta (x-x^*)+\psi_2(x^*), \ \   x\geq x^*,
\end{array}\right.
\end{eqnarray}
where $x_0
 =\frac{(1-\gamma)\sigma^2}{\mu}$.
  $x^*$, $\gamma$, $d_-$ and $d_+$ are solutions of (\ref{x^*}), (\ref{e310}) and
(\ref{e315}), and $C_1$, $C_2$ and $C_3$ are determined by
(\ref{C1}),(\ref{C3}) and (\ref{C4}), respectively.
 \vskip 10pt\noindent
The problem remained is to approve the following inequalities.\\ For
$0\leq x\leq x^*$,
\begin{eqnarray}\label{e212}
-\frac{\partial\phi}{\partial x}(s,x)+\beta e^{-cs}<0,
\end{eqnarray}
\begin{eqnarray}\label{e209}
\max\limits_{a\in[0,1]}\big\{\frac{\partial\phi}{\partial
s}&+&\frac{\partial\phi}{\partial
x}a\mu+\frac{1}{2}a^2\sigma^2\frac{\partial^2\phi}{\partial x^2}
+\int_{\Re}\{\phi(s,x+akz)-\phi(s,x)\nonumber\\&-&akz\frac{\partial\phi}{\partial
x}(s,x)\}\nu(dz)\big\}\leq0.
\end{eqnarray}
For $x\geq x^*$,
\begin{eqnarray}\label{e210}
&-&\frac{\partial\phi}{\partial x}(s,x)+\beta e^{-cs}=0,
\end{eqnarray}
\begin{eqnarray}\label{e211}
\max\limits_{a\in[0,1]}\big\{\frac{\partial\phi}{\partial
s}&+&\frac{\partial\phi}{\partial
x}a\mu+\frac{1}{2}a^2\sigma^2\frac{\partial^2\phi}{\partial x^2}
+\int_{\Re}\{\phi(s,x+akz)-\phi(s,x)\nonumber\\&-&akz\frac{\partial\phi}{\partial
x}(s,x)\}\nu(dz)\big\}\leq0.
\end{eqnarray}
Since
\begin{eqnarray*}
\phi_1^{''}(x)=e^{-cs}C_1\gamma(\gamma-1)x^{\gamma-2}<0,\\
\phi_2^{''}(x)=e^{-cs}\frac{\beta
d_{+}d_{-}}{d_{+}-d_{-}}(e^{d_{-}(x-x^*)}-e^{d_{+}(x-x^*)})<0
\end{eqnarray*}
for $\gamma<1$ and $x\leq x^*$, the inequality (\ref{e212}) is
trivial due to $\phi(x)\in C^2$ is a convex function, and the
inequality (\ref{e210}) is a direct consequence of
$\phi_3''(x)=\beta$ for $x\geq x^*$ . \vskip 10pt\noindent
 For $0\leq x\leq x_0$, by the expression  of $\phi$,
 $ \max\limits_{a\in[0,1]}\{\mathcal {A}\phi\}=0$ is obvious.
 \vskip 10pt\noindent
For $x_0\leq x\leq x^*$, the inequality (\ref{e209}) is equal to
\begin{eqnarray}\label{e203}
\max\limits_{a\in[0,1]}\{\frac{1}{2}a^2\sigma^2\psi_2^{''}(x)
&+&a\mu \psi_2^{'}(x)-c\psi_2(x)
+\int_{\Re}\{\psi_2(x+akz)-\psi_2(x)\nonumber\\&-&akz\psi_2^{'}(x)
\}\nu(dz)\}\leq0.
\end{eqnarray}
Denote the function in bracket $\{\cdot\}$ at the left side of the
inequality (\ref{e203}) as $ p(a)$, we will prove that $p(a)$ is an
increasing function of $a$.
\begin{eqnarray*}
p^{'}(a)&=&a\sigma^2\psi^{''}_2(x)+\mu\psi^{'}_2(x)
=a\sigma^2[C_3(d_{-})^2e^{d_{-}x}+C_4(d_{+})^2e^{d_{+}x}]
\nonumber\\&+&\mu[C_3d_{-}e^{d_{-}x}+C_4d_{+}e^{d_{+}x}]
=\frac{\beta
d_{+}d_{-}}{d_{+}-d_{-}}(e^{d_{-}(x-x^*)}-e^{d_{+}(x-x^*)})
\nonumber\\&+&\mu[C_3d_{-}e^{d_{-}x}+C_4d_{+}e^{d_{+}x}]\geq0.
\end{eqnarray*}
as $d_{-}<0$, $d_{+}>0$, $x\leq x^*$, $ C_3<0$, and $C_4>0$. Then
$p(a)\leq p(1)=0$ for $0\leq a\leq 1$. \vskip 10pt\noindent For
$x\geq x^*$, the inequality(\ref{e211}) is equal to
\begin{eqnarray*}
\max\limits_{a\in[0,1]}\frac{1}{2}a^2\sigma^2\psi_3^{''}(x)& +&a\mu
\psi_3^{'}(x)-c\psi_3(x)
+\int_{\Re}\{\psi_3(x+akz)-\psi_3(x)\nonumber\\&-&akz\psi_3^{'}(x)
\}\nu(dz) =a\mu\beta-c\beta(x-x^*)-c\psi_2(x^*)\nonumber\\ &\leq&
\mu\beta-c\psi_2(x^*)-c\beta(x-x^*)\leq0
\end{eqnarray*}
due to $x\geq x^*$  and $\mu\beta-c\psi_2(x^*)=0$. So we end the
proof.
 \end{proof}
 \vskip 1pt\noindent
\setcounter{equation}{0}
\section{{ {\bf The solution of the optimal control problem with jump diffusions}} }
\vskip 10pt\noindent We now give a verification theorem for singular
-regular control problem(\ref{e2.1}) and (\ref{e217}).  We first
prove  the following.
\begin{Them} \label{T41}
Let W(s, x) satisfy  the following HJB equation,
\begin{eqnarray}\label{e322}
&&\max\{-\frac{\partial W}{\partial x}(t,x)+\beta e^{-ct},
\max\limits_{a\in[0,1]}\{\mathcal {A} W(t,x)\} \}=0\\
&&  \ \mbox{ for
$t\geq 0$ and $ x\geq 0$,}\nonumber\\
&& W(t, 0)=0\ \mbox{ for any $t\geq 0$}.
\end{eqnarray}
Then $W(s,x)\geq J(s, x,\pi)$ for any admissible strategy $\pi $ and
$(s,x)\in \Re^2_+$.
\end{Them}
\begin{proof} For any  fixed  strategy $\pi$, let
$\Lambda=\{s:L_{s-}^{\pi}\neq L_{s}^{\pi}\}$,
$\hat{L}=\sum_{s\in\Lambda, s\leq t}(L_{s}^{\pi}-L_{s-}^{\pi})$ be
the discontinuous part of $L_{s}^{\pi}  $ and
$\tilde{L}_{t}^{\pi}=L_{t}^{\pi}-\hat{L}_{t}^{\pi}$ be the
continuous part of $L_{s}^{\pi}  $. Let $\tau_\pi$ be the first time
that the corresponding cash  flow $R^\pi_{t}$ defined by (2.2) hit
$(-\infty, 0)$. Then, by applying the generalized It\^{o} formula to
the stochastic process $ Y^\pi(t):=( s+t, R^\pi_t)^T$ and the
function $W(s,x)$, we have
\begin{eqnarray}\label{e213}
&&\mathbf{E} [W(s+t\wedge\tau_\pi,
R_{t\wedge\tau_\pi}^{\pi})]\nonumber\\&&= W(s,
x)+\mathbf{E}[\int_{0}^{t\wedge \tau_\pi}\mathcal
{A} W(s+u,R_{u}^{\pi})du\nonumber\\
&&-\int_{0}^{t\wedge \tau_\pi}\frac{\partial W(s+u,
R^\pi_{u})}{\partial x}dL^{(c)}_{u}+\sum\limits_{0<t_n\leq
t\wedge\tau_\pi}\Delta_LW(s+t_n,
R^\pi_{t_n})],\nonumber\\
\end{eqnarray}
where
\begin{eqnarray*}
\mathcal {A}W(s,x)&=&\frac{\partial W}{\partial
s}+a\mu\frac{\partial W}{\partial
x}+\frac{1}{2}a^2\sigma^2\frac{\partial^2W}{\partial x^2}
+\int_{\Re}\{W(s,x+akz)-W(s,x)\nonumber\\&-&akz\frac{\partial
W}{\partial x}(s,x)\}\nu(dz),
\end{eqnarray*}
\begin{eqnarray*}
&&\Delta_LW(s+t_n, R^\pi_{t_n}):= W(Y^\pi(t_n) )- W(Y^\pi(t^-_n)
+\Delta_NY^\pi(t_n) ),\\
&&\Delta_NY^\pi(t_n) ):=\big(0, ka_\pi(t_n)\int_\Re
z\widetilde{N}(\{t_n\},dz)\big ),\\
&&\{ t_k\} \mbox{ is the set of jumping times of $L$    }.
\end{eqnarray*}
 Using $\mathcal {A} W\leq 0$ in the equation (\ref{e213}), we see that
\begin{eqnarray}\label{e214}
&&\mathbf{E}[W(s+t\wedge \tau_\pi, R_{t\wedge \tau_\pi}^{\pi})]\leq
W(s,
x)\nonumber\\
&&-\mathbf{E}[\int_{0}^{t\wedge\tau_\pi}\frac{\partial W(s+u,
R^\pi_{u})}{\partial x}dL^{(c)}_{u}- \sum\limits_{s<t_n\leq
t\wedge\tau_\pi}\Delta_LW(Y^\pi_{t_n})].
\end{eqnarray}
By the mean value theorem we have
\begin{eqnarray*}
\Delta_L W(Y^\pi_{t_n})=-\frac{\partial W}{\partial
x}(\hat{Y}^{(n)}_{t_n})\Delta L(t_n),
\end{eqnarray*}
where $\hat{Y}^{(n)}_{t_n}$ is some point on the straight line
between $Y^\pi_{t_n}$ and $Y^\pi_{t_n^-}+\Delta_N(Y^\pi_{t_n}  )$.
Since $W^{'}(Y_{u}^{\pi})\geq \beta e^{-c(s+u)}$,
$$\Delta_L W(Y^\pi_{t_n})\leq
-\beta e^{-c(s+t_n)}(L_{t_n}^{\pi}-L_{t_n-}^{\pi}),$$ which,
together with the inequality (\ref{e214}), implies
\begin{eqnarray}\label{e215}
\mathbf{E} [W(s+t\wedge \tau_\pi, R_{t\wedge
\tau_\pi}^{\pi})]&+&\mathbf{E}\big\{\int_{0}^{ t\wedge\tau_\pi}\beta
e^{-c(s+u)}dL_{u}^{\pi}\big \}\leq W(s, x).\nonumber\\
\end{eqnarray}
By the definition of $\tau_\pi $, the boundary condition (4.2) and
$W^{'}(Y_{u}^{\pi})\geq \beta e^{-c(s+u)}$, it is easy to prove that
$\liminf\limits_{t\rightarrow\infty}W(Y_{t})I_{\{\tau_\pi=\infty\}}=0
$ and
\begin{eqnarray}\label{e216}
\liminf\limits_{t\rightarrow\infty}W(s+t\wedge \tau_\pi, R_{t\wedge
\tau_\pi}^{\pi})&=&W(s+\tau_\pi,0)I_{\{\tau_\pi<\infty\}}+
\liminf\limits_{t\rightarrow\infty}W(Y_{t})I_{\{\tau_\pi=\infty\}}\nonumber\\
&\geq& W(s+\tau_\pi,0)I_{\{\tau_\pi<\infty\}}= 0.
\end{eqnarray}
So, we deduce from the inequalities (\ref{e215}) and (\ref{e216})
that
\begin{eqnarray*}
J(s, x,\pi)=\mathbf{E}\big\{\int_{0}^{\tau_\pi} e^{-c(s+t)}\beta
dL_{t}^{\pi}\big\}\leq W(s,x),
\end{eqnarray*}
thus we complete the proof.
\end{proof}
 \vskip 10pt\noindent
  Let
\begin{eqnarray*}
a(x)=\left\{
\begin{array}{l l l}
\frac {\mu
x}{\sigma^2(1-\gamma )},\quad x< x_0,\\
 1,\qquad \quad  x\geq x_0
\end{array}
\right.
\end{eqnarray*}
where $x_{0}=\frac{(1-\gamma)\sigma^{2}}{\mu}$. We call $a(x)$ the
feedback control function of the control problem (\ref{e2.1}) and
(\ref{e217}). \vskip 5pt\noindent We can now state the main result
of this paper.
\begin{Them}\label{Them42}
Assume that (\ref{relat})holds, the L\'{e}vy measure $\nu$  and the
adjusted risk rate $k$ satisfy $0<  \nu( \Re)< +\infty$, $0<\int_\Re
z\nu(dz)< +\infty  $ and $ 0< k \leq \frac{\mu}{2\int_\Re z\nu(dz)
}$. Then the optimal return function and the optimal dividend
strategy of the control problem (\ref{e2.1}) and (\ref{e217}) are
$V(s,x)=\phi(s,x)=e^{-cs}\psi(x)$ and $\pi^{*}= ( a( R_{t}^{\pi^*}),
L^{\pi^*}_t) $, respectively, where $ ( R_{t}^{\pi^*}, L^{\pi^*}_t
)$ is uniquely determined by the following stochastic differential
equations with reflection,
\begin{eqnarray}\label{e42*}
\left\{
\begin{array}{l l l}
R_{t}^{\pi^*}=x+\int_{0}^{t}\mu
a(R_{s}^{\pi^*})ds+\int_{0}^{t}\sigma a(R_{s}^{\pi^*})d{W}_{s}+k\int_{0}^{t}\int_\Re a(R_{s}^{\pi^*}) z\widetilde{N}(ds,dz)\\-L_{t}^{\pi^*},\\
R_{t}^{\pi^*}\leq x^{*},\\
\int_{0}^{\infty}I_{\{t:R_{t}^{\pi^*}<x^{*}\}}(t)dL_{t}^{\pi^*}=0,
 \end{array}\right.
\end{eqnarray}
$ \psi(x)  $  is the function defined by (\ref{eq200}) and the
optimal dividend payments level $x^{*}$ is given  by (\ref{x^*}).
\end{Them}
\begin{proof} Since  the function  $\phi(s, x)$ satisfies the HJB
equations (\ref{e32}), it is not hard to see that $\phi(s, x)$ also
satisfies conditions  in Theorem \ref{T41}. So $\phi(s, x)\geq J(s,
x,\pi)$  for any $  \pi$,  i.e.,
\begin{eqnarray}\label{e48*}
 \phi(s,x)\geq V(s,x).
\end{eqnarray}
 \vskip
10pt\noindent Next, we will prove $V(s,x)=\phi(s,x)= J(s, x,\pi^*)$
corresponding to $\pi^{*}$. By applying the generalized It\^{o}
formula, noting that the construction of $\phi(s,x)$ and the last
two equations in (\ref{e42*}), we deduce from the
inequality(\ref{e212})and the equations (\ref{e322}) that $ \mathcal
{A}\phi(s+t, R_{t}^{\pi^{*}})=0 $  for any $t\geq 0$,
 $\int_{0}^{t\wedge\tau^*}\frac{\partial
\phi(Y^{\pi^{*}}_{u})}{\partial x}dL^{(c)}_{u}= \int_{0}^{
t\wedge\tau^*}\beta e^{-c(s+u)}dL_{u}^{(c)}    $ and
 $ \sum\limits_{s<t_n\leq\tau^*}
\Delta_L\phi(Y^{\pi^{*}}_{t_n})=-\sum\limits_{s<t_n\leq\tau^*}\frac{\partial
\phi}{\partial x}(s+t_n, x^*)\Delta L(t_n)=
-\sum\limits_{s<t_n\leq\tau^*}\beta e^{-c(s+t_n)}\Delta L(t_n) $,
where $\tau^*=\inf\{t\geq0: R^{\pi^*}_{t}<0\}$. So
\begin{eqnarray}\label{e49*}
\mathbf{E} [\phi(s+t\wedge \tau^*, R_{t\wedge \tau^*}^{\pi^*})]&=
&\phi(s, x)+\mathbf{E}[\int_{0}^{t\wedge\tau^*}\mathcal
{A} \phi(Y_{u}^{\pi^*})du\nonumber\\
&-&\int_{0}^{t\wedge\tau^*}\frac{\partial
\phi(Y^{\pi^{*}}_{u})}{\partial
x}dL^{(c)}_{u}+\sum\limits_{s<t_n\leq\tau^*}
\Delta_L\phi(Y^{\pi^{*}}_{t_n})]\nonumber\\
&=&\phi(s, x)-\mathbf{E}[\int_{0}^{ t\wedge\tau^*}\beta
e^{-c(s+u)}dL_{u}^{(c)} \nonumber\\
&&+\sum\limits_{s<t_n\leq\tau^*}\beta
e^{-c(s+t_n)}\Delta L(t_n)] \nonumber\\
&=&\phi(s, x)-\mathbf{E}[\int_{0}^{ t\wedge\tau^*}\beta
e^{-c(s+u)}dL_{u}^{\pi^{*}}].
\end{eqnarray}
 \vskip 10pt\noindent
Since $\lim\limits_{t\rightarrow\infty} \phi(s+t\wedge\tau^{*},
 R_{t\wedge\tau^{*}}^{\pi^{*}})=\lim\limits_{t\rightarrow\infty}
e^{-c(s+t\wedge\tau^{*})}
\psi(R_{t\wedge\tau^{*}}^{\pi^{*}})=e^{-c(s+\tau^{*})}\psi(R_{\tau^{*}}^{\pi^{*}})
=e^{-c(s+\tau^{*})}\psi(0) =0$, we see from the inequality
(\ref{e48*}) and the equation(\ref{e49*})
that
$$V(s,x)\leq \phi(s,x)= \lim\limits_{ t\rightarrow\infty} \mathbf{E}\big
[\int_{0}^{t\wedge \tau^{*}}e^{-cs}\beta_{1}dL^{\pi^{*}}_{s}\big
]=J(s, x,\pi^{*})\leq V(s,x).$$
 So $V(s,x)=\phi(s,x)=J(s, x,\pi^{*})$, that is, $\phi(s,x)    $
  is the optimal  return function,  $\pi^{*}$ is the  optimal dividend
  strategy and  $x^*$ is the optimal dividend payments level.
Thus the proof has been done.
\end{proof}
\vskip 20pt\noindent
 \setcounter{equation}{0}
\section{{ {\bf Numerical examples }} }
 \vskip 15pt\noindent
 In this section, based on Theorem \ref{Them42}, we present some
  numerical examples, together with the feedback control function
  $a(x)$ and the
comparison theorem for SDE,  to portray how the key model
parameters(e.g. $k$,  $\mu$,  $\sigma^2$ and $\nu$) impact on
$V(s,x)$ and the optimal control strategy $\pi^* $,  that is, $
a_{\pi^*} (t)$ and $L^{\pi^*}_{t}$, respectively.
 \begin{example}\label{ex51}
Let $\nu(dz)=e^{-z}I_{\{z\geq 0\}}(z)dz$. Figure \ref{2} below
explains that the adjusted risk rate will increase the optimal
dividend payments level $x^*(k )$, so to avoid bankruptcy the
company should decrease the times of dividend or increase $k$ if
possible, that is, the company needs to maintain the cash inside the
company to cover the catastrophe risk ,so it pays dividend at a
higher level. On the other hand, $L^{\pi^*}_{t}$ decreases with $k$
by (\ref{e42*}), $R_{t}^{\pi^*} $ increases with $k$, so we see that
$ a_{\pi^*} (t)$  also increases with $k$. In fact,  the catastrophe
risk business brings in more risk as well as more income, and the
higher asset level raises the risk sustainment of the company. It
could reduce its reinsurance level $( i.e.,  1-a_{\pi^*} (t) ) $
according with the optimal control strategy $ \pi^* $.
\begin{figure}[H]
\includegraphics[width=0.7\textwidth]{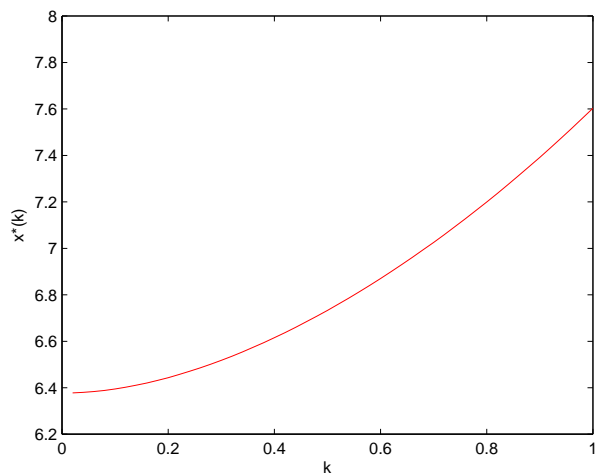}
\caption{ The optimal dividend payments level $x^*(k )$ as a
function of $k$.The parameter values are $\sigma^2=5$, $c=0.05$,
$\mu =2$, $s=0$.}\label{2}
\end{figure}
\end{example}
\begin{example}\label{ex52} Let $\nu(dz)=e^{-z}I_{\{z\geq 0\}}(z)dz$.
Figure \ref{1} below states that the property insurance company's
profit increases with the initial capital $x$ and the adjusted risk
rate $k$. So the property insurance company can get some return from
its catastrophe risk insurance business, but the return's increment
is small by adjusting $k$. However, the company can receive a good
public reputation by constant $k$, and interest from the catastrophe
insurance business.
\begin{figure}[H]
\includegraphics[width=0.8\textwidth]{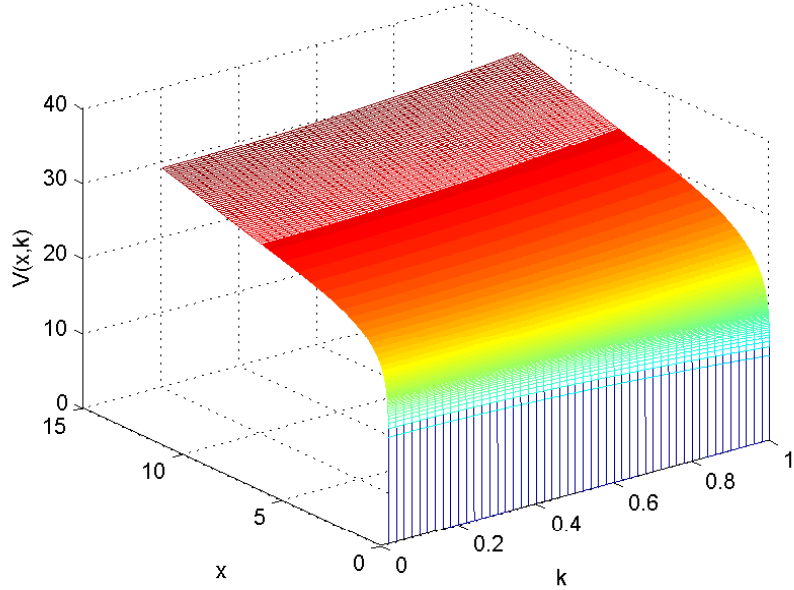}\vskip 10pt
\caption{ The optimal return function $V(x, k )$ as a function of $
x$ and $k$. The parameter values are $\sigma^2=5$, $c=0.05$,
$\beta=0.8$, $s=0$, $\mu=2$. }\label{1}
\end{figure}
\end{example}
\begin{example}\label{ex53}
Let $\nu(dz)=e^{-z}I_{\{z\geq 0\}}(z)$. Figure \ref{4} below
portrays that the optimal dividend payments level $x^*(\mu)$
decreases with the premium rate $\mu$, so $L^{\pi^*}_{t}$ increases
with $\mu$, but $ a_{\pi^*} (t)$ decreases with the premium rate.
These facts  mean that the higher growth rate of the insurance
company's asset raise the company's risk tolerance level and the
company could pay dividend at a lower level. Meanwhile, the company
should  adopt a higher reinsurance rate to avoid  bankruptcy due to
the lower dividend payments level.
\begin{figure}[H]
\includegraphics[width=0.7\textwidth]{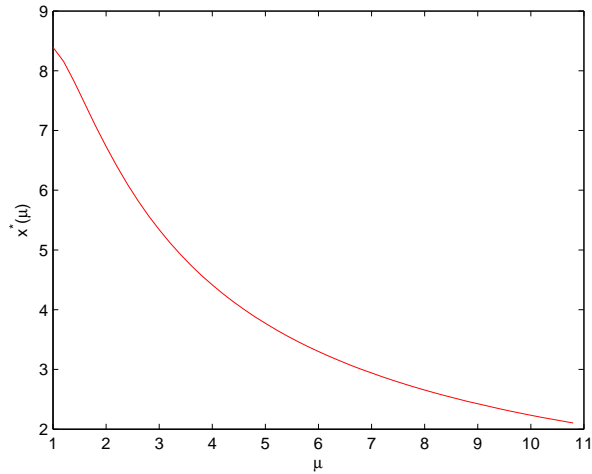}
\caption{ The optimal dividend payments level $x^*(\mu )$ as a
function of $\mu$. The parameter values are $\sigma^2=5$, $c=0.05$,
$s=0$, $k=0.5$.}\label{4}
\end{figure}
\end{example}
\vskip 10pt
\begin{example}\label{ex54}
Let $\nu(dz)=e^{-z}I_{\{z\geq 0\}}(z)dz$. Figure \ref{3} states that
the optimal return function $V(x, \mu )$ is an increasing function
of $\mu$, and high premium rate can notably increase the company's
return, that is, a higher growth rate of  the insurance company's
asset results in a higher return.
 \begin{figure}[H]
\includegraphics[width=0.7\textwidth]{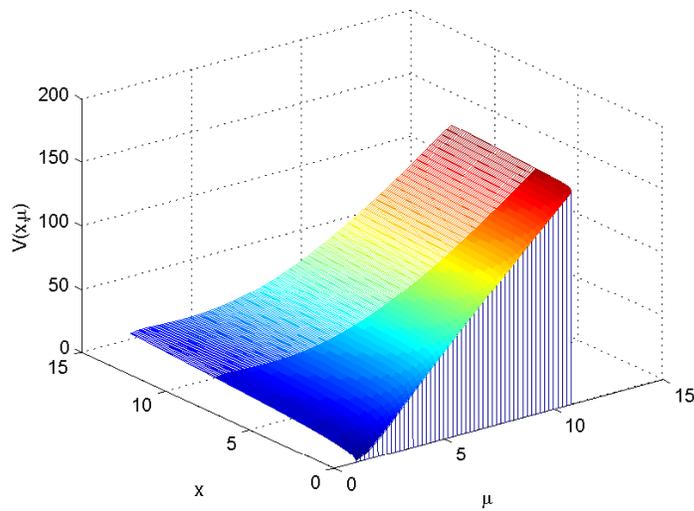}
\caption{ The optimal return function $V(x, \mu )$ as a function of
$ x$ and $\mu$. The parameter values are  $\sigma^2=5$, $c=0.05$,
$\beta=0.8$, $s=0$, $k=0.5$.}\label{3}
\end{figure}
\end{example}
\vskip 15pt\noindent
\begin{example}\label{ex55}
Let $\nu(dz)=e^{-z}I_{\{z\geq 0\}}(z)dz$. Figure \ref{6} below
portrays that the optimal dividend payments level $x^*(\sigma^2)$
increases with the risk volatility rate $\sigma^2$ of normal
insurance business, so $L^{\pi^*}_{t}$ decreases with $\sigma^2$,
but $ a_{\pi^*} (t)$ increases it. These mean that the higher
volatility make the insurance company's asset reduce the company's
risk tolerance level and the company prefer to maintain the cash
inside the company to cover the risk. Meanwhile, the company should
adopt a lower reinsurance rate to get lower optimal dividend
payments level.
\begin{figure}[H]
\includegraphics[width=0.7\textwidth]{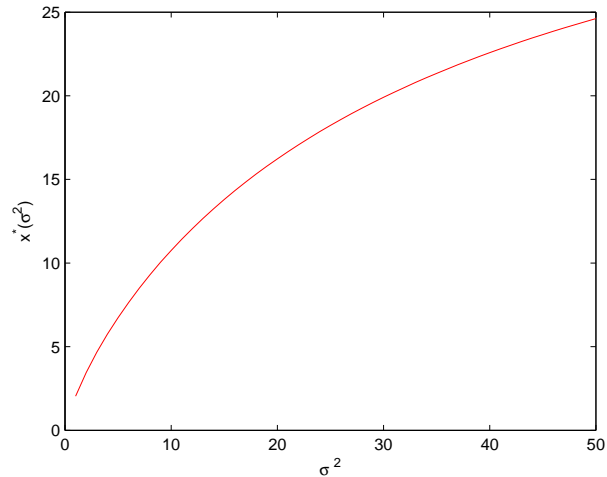}
\caption{ The optimal dividend payments level $x^*(\sigma^2)$ as a
function of $\sigma^2$. The parameter values are  $c=0.05$, $s=0$,
$k=0.5$, $\mu=2$.}\label{6}
\end{figure}
\end{example}
\begin{example}\label{ex56}
Let $\nu(dz)=e^{-z}I_{\{z\geq 0\}}(z)dz$. Figure \ref{5} below
states that the increment of the optimal return function $V(x,
\sigma^2 )$ due to $\sigma^2$ is very large, so higher risk can also
notably increase the company's return.
\end{example}
 \begin{figure}[H]
\includegraphics[width=0.7\textwidth]{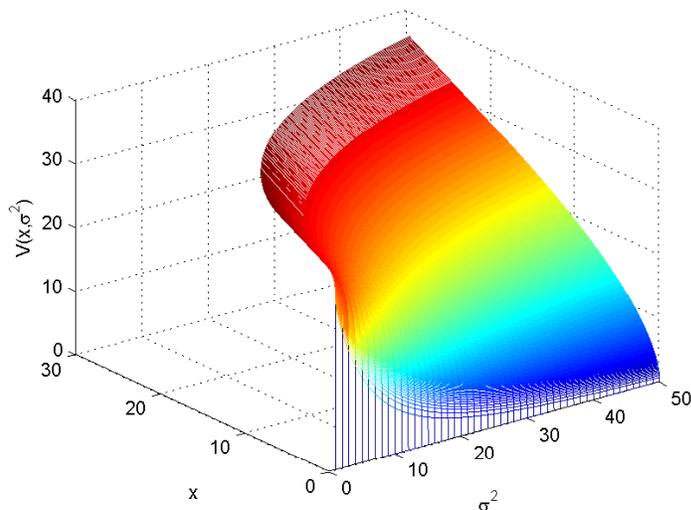}
\caption{ The optimal return function $V(x, \sigma^2 )$ as a
function of $ x$ and $\sigma^2$. The parameter values are $c=0.05$,
$\beta=0.8$, $s=0$, $k=0.5$, $\mu=2$.}\label{5}
\end{figure}
\begin{example}\label{ex57}
Let $\nu_t(dz)=e^{-tz}I_{\{z\geq 0\}}(z)dz(t\geq 1)$. Figure \ref{8}
below portrays that the optimal dividend payments level $x^*(t)$ has
obvious decrements on $[1,4]$, but on $[4, +\infty )$ the optimal
dividend payments level has no visibly changes, so $ a_{\pi^*}
(\cdot)$ and $L^{\pi^*}_{\cdot}$ change greatly for different
L\'{e}vy measures $\nu_t(dz)=e^{-tz}I_{\{z\geq 0\}}(z)dz$, $t\in
[1,4]$. However, they are almost same for different L\'{e}vy
measures $\nu_t(dz)=e^{-tz}I_{\{z\geq 0\}}(z)dz$, $t\geq 4$.
\begin{figure}[H]
\includegraphics[width=0.7\textwidth]{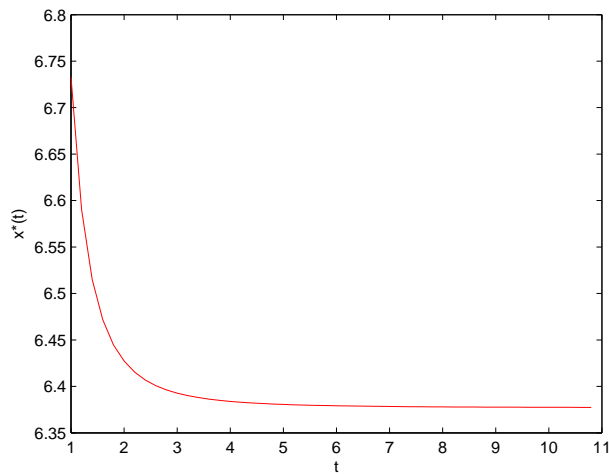}
\caption{ The optimal dividend payments level $x^*(t)$ as a function
of $t$. The parameter values are $\sigma^2=5$, $c=0.05$, $s=0$,
$k=0.5$, $\mu=2$.}\label{8}
\end{figure}
\end{example}
\begin{example}\label{ex58}
Let $\nu_t(dz)=e^{-tz}I_{\{z\geq 0\}}(z)dz$. Figure \ref{7} below
states that the change of the optimal return function $V(x, t )$ for
different $t$ is not distinct. So the optimal return function $V(x,
t )$ is nearly stable for different L\'{e}vy measures
$\nu_t(dz)=e^{-tz}I_{\{z\geq 0\}}(z)dz$ ($t\geq 0$).
\begin{figure}[H]
\includegraphics[width=0.7\textwidth]{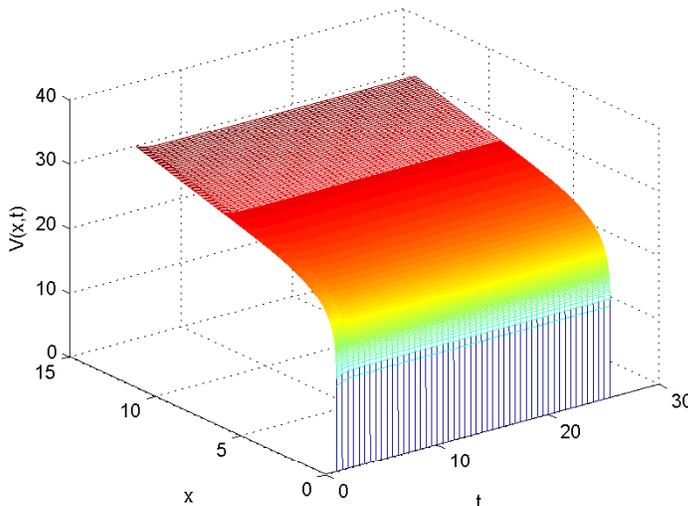}
\caption{ The optimal return function $V(x,t )$ as a function of $
x$ and $t$. The parameter values are $\sigma^2=5  $, $c=0.05$,
$\beta=0.8$, $s=0$, $k=0.5$, $\mu=2$.  )}\label{7}
\end{figure}
\end{example}
 \setcounter{equation}{0}
\section{{ {\bf Conclusion  }} }
\vskip 15pt\noindent
 We consider the optimal dividend and the reinsurance strategy of
 a property insurance company. The property insurance
business brings in catastrophe risk, such as earthquake and flood.
The catastrophe risk could be partly reduced by reinsurance. Due to
the huge risk, the company needs to add a  adjusted risk rate in the
regulation. The management of the company controls the reinsurance
rate and dividend payments to maximize the expected present value of
the dividends before bankruptcy. This is the first time to consider
the catastrophe risk in an insurance model, which is more realistic.
The catastrophe risk is modeled as the jump process in the
stochastic control problem. In order to find the solution of the
problem, we implore the mixed singular-regular control methods of
jump diffusions. We establish the optimal reinsurance rate, the
optimal dividend strategies and explicit the optimal return function
of the company. The influences of the catastrophe risk and the
reinsurance regulation of the catastrophe risk on the optimal
control strategy of the insurance company are also discussed. Based
on the main results we have just established, we present some
numerical examples to analyze in detail how the key model parameters
impact on the optimal retention ratio, the optimal dividend payments
strategies and the optimal return of the company. \vskip
40pt\noindent
 {\bf Acknowledgements.} \ \
This work is supported by Projects 10771114 and 11071136 of NSFC,
Project 20060003001 of SRFDP, the SRF for ROCS, SEM  and the Korea
Foundation for Advanced Studies.   We would like to thank the
institutions for the generous financial support. Special thanks also
go to the participants of the seminar stochastic analysis and
finance at Tsinghua University for their feedbacks and useful
conversations. Li He also thanks the generous financial support of
10XNF057 from Renmin University of China. We are also grateful to
Dr. Huan Fan and Yifeng Yin for useful discussions. \vskip
10pt\noindent
 \setcounter{equation}{0}


\begin{thebibliography}{99}
\baselineskip 17pt
\bibitem{AT}Hansj\"{o}rg Albrecher, Stefan Thonhauser: Optimality
results for dividend problems in insurance, {\bf  RACSAM, Rev. R.
Acad. Cien. Serie A. Mat.  Vol.103}, 295-320, 2009.
\bibitem{s8} Asmussen, S., H{\o}jgaard, B., Taksar, M.:
Optimal Risk Control and Dividend Distribution Policies: Example of
Excess-of-Loss Reinsurance for an insurance corporation, {\bf
Finance and Stochastics, Vol. 4}, 199-324, 2000.
\bibitem{AV} Avanzi, B., 2009. Strategies for Dividend Distribution:
A Review. {\bf North American Actuarial Journal, Vol. 13}, No. 2,
pp. 217-251.
\bibitem{s7} Asmussen, S., Taksar,
M.: Controlled Diffusion Models for Optimal Dividend Pay-out, {\bf
Insurance: Mathematics and Economics, Vol.20}, 1-15, 1997.
\bibitem{s202}Bernt {\O}ksendal, Agn\`{e}s Sulem: Applied Stochstic
Control of Jump Diffusions, ISBN:3540140239, Springer-Verlag Berlin
Heidelberg, 2005.
\bibitem{s10} Borch, K.: The Capital Structure of a
Firm, {\bf The Swedish Journal of Economics, Vol.71}, 1-13, 1969.
\bibitem{s9} Borch, K.: The Theory of Risk, {\bf Journal of
the Royal Statistical Society. Series
, B 29}, 432-452, 1967.
\bibitem{BRJS} Basse,T.,Reddemann,S., Riegler, J.J., Schulenburg, J. :
Dividend strategy and the global financial crisis: empirical
evidence from the Italian insurance industry, {\bf Zeitschrift f¨¹r
die gesamte Versicherungswissenschaft(German Journal of Risk and
Insurance), Volume 97}, Supplement 1, 155-171, 2010.
\bibitem{s6} Cadenillas, A., Choulli, T., Taksar M., Zhang Lei:
Classical and Impulse Stochastic Control for the Optimization of the
Dividend and Risk Policies of an Insurance Firm, {\bf  Mathematical
Finance, Vol. 16}, No. 1, 181-202, January 2006.
\bibitem{s11} Gerber, H. U.:
Games of Ecomonic Survival with Discrete and Continous Income
Processes, {\bf Operations Research,  Vol.20}, 37-45, 1972.
\bibitem{s1} GUO Xin, LIU Jun,
ZHOU Xunyu: A Constrained Nonlinear Regular-singular Stochastic
Control Problem, {\bf Stochastic Processes and Their Applications,
Vol.109}, 167-187, 2003.
\bibitem{s101} Harrison, J.M.; Taksar, M.J.: Instant control of
Brownian motion, {\bf Mathematics of Operations Research}, {\bf
8}(1983)439-453.
\bibitem{s200} He, Lin and Liang, Zongxia:  Optimal Financing and
Dividend Control of the Insurance Company with Proportional
Reinsurance Policy.  {\bf Insurance: Mathematics and Economics, Vol.
42}, Issue 3, 976-983, 2008.
\bibitem{s201} He, Lin and Liang, Zongxia: Optimal Financing and
Dividend Control of the Insurance Company with Fixed and
Proportional Transaction Costs. {\bf Insurance£ºMathematics and
Economics, Vol. 44}, 88-94, 2009.
\bibitem{ime02} Lin He, Zongxia Liang, 2008.
 Optimal Dividend Control of the Insurance Company with
Proportional Reinsurance Policy under solvency constraints. {\bf
Insurance: Mathematics and Economics, Vol.43}, 474-479.
\bibitem{s12} H{\o}jgaard, B.,
Taksar, M.: Optimal Proportional Reinsurance Policies for Diffusion
Models, {\bf  Scandinavian Actuarial Journal, Vol. 2}, 166-180,
1998.
\bibitem{s3} H{\o}jgaard,
B., Taksar, M.: Controlling Risk Exposure and Dividends Payments
Schemes: Insurance company Example, {\bf  Mathematical Finance,
Vol.} 9, No. 2, 153-182, April 1999.
\bibitem{s13}H{\o}jgaard, B., Taksar, M.:
Optimal Proportional Reinsurance Policies with Transaction Costs,
 {\bf Insurance: Mathematics and Economics, Vol. 22}, 41-51, 1998.
\bibitem{s4}
H{\o}jgaard, B., Taksar, M.: Optimal Risk Control for a Large
Corporation in the Presence of Returns on Investments, {\bf Finance
and Stochastics, Vol. 5}, 527-547, 2001.
\bibitem{s14} Martin-L\"{o}f, A.:
Premium Control in an Insurance System, an Approach Using Linear
Control Theory, {\bf Scandinavian Actuarial Journal}, 1-27, 1983.
\bibitem{s15} Paulsen, J., Gjessing, H. K.:
Optimal Choice of Dividend Barriers for a Risk Process with
Stochastic Return of Investment, \  {\bf  Insurance: Mathematics and
Economics, Vol. 20}, 215-223, 1997.
\bibitem{s18}Radner, R., Sheep, L.:
Risk vs. Profit Potential: A Model for Corporate Strategy, {\bf
Journal of Economic Dynamics and Control, Vol. 20}, 1373-1393, 1996.
\bibitem{s16} Taksar, M.:
Optimal Risk/dividend Distribution Control Models: Applications to
Insurance, {\bf Mathematical Methods of Operations Research,
Vol.51}, Number 1, 1-42 2000.
\end{thebibliography}
\end{document}